\documentclass[11pt]{article}
\usepackage{amsmath, mathtools}
\usepackage{amsthm}
\usepackage{amssymb}
\usepackage{graphicx}
\usepackage{appendix}
\usepackage{color}

\usepackage{hyperref}

\usepackage{float}

\usepackage{multirow}

\newtheorem{theorem}{Theorem}[section]

\numberwithin{equation}{section}

\begin{document}

\title{Refinements of Barndorff-Nielsen and Shephard model:  an analysis of crude oil price with machine learning}
\author{Indranil SenGupta\footnote{Associate Professor and Graduate Program Director, Department of Mathematics, North Dakota State University, Fargo, North Dakota, USA. Email: indranil.sengupta@ndsu.edu}, William Nganje\footnote{Professor and Department Chair, Department of Agribusiness and Applied Economics at North Dakota State University, Fargo, North Dakota, USA.},  Erik Hanson\footnote{Assistant Professor, Department of Agribusiness and Applied Economics at North Dakota State University, Fargo, North Dakota, USA.}}

\date{\today}

\maketitle

\begin{abstract}

A commonly used stochastic model for derivative and commodity market analysis is the Barndorff-Nielsen and Shephard (BN-S) model. Though this model is very efficient and analytically tractable, it suffers from the absence of long range dependence and many other issues. For this paper, the analysis is restricted to crude oil price dynamics.
A simple way of improving the BN-S model with the implementation of various machine learning algorithms is proposed. This refined BN-S model is more efficient and has fewer parameters than other models which are used in practice as improvements of the BN-S model.  The procedure and the model show the application of data science for extracting a ``deterministic component" out of processes that are usually considered to be completely stochastic. Empirical applications validate the efficacy of the proposed model for long range dependence.

\end{abstract}
\section{Introduction}
\label{sec1}

One of the most prominent tools in modern big data analysis is machine learning. Machine learning is about extracting knowledge from a significantly large data set. The application of machine learning methods has recently become ubiquitous in everyday life. Machine learning has had a tremendous influence on the way data-driven research is done today. The tools can be applied to diverse scientific problems such as understanding stars, finding distant planets, discovering new particles, analyzing DNA sequences, and providing personalized cancer treatments.

A commodity of fundamental importance is crude oil. Consequently an analysis of the dynamics of crude oil price time series seems to be crucial. This allows to ascertain the potential impacts of its shocks in several economies and on other financial assets (see \cite{Tabak}). As observed in \cite{Sensoy}, long-range dependence is evident in various energy futures markets. 
Many other existing works are dedicated to the analysis of the dynamics of crude oil prices. In \cite{Frey}, various econometric models used to forecast crude oil prices are summarized and interpreted. In \cite{oil11}, a deep learning model is applied to crude oil prices and a hybrid crude oil price forecasting model is provided. In  \cite{Sircar1}, oil producers' decisions in Cournot competitions are described through continuum dynamic mean field games. In related work (see \cite{Sircar2}), a modified Hotelling's rule for games with stochastic demand is discussed. In \cite{Michael}, machine learning algorithms are implemented to analyze the oil price dynamics for the Bakken region in the United States.

Paper \cite{Li} uses a convolutional neural network to forecast crude oil prices through online media text mining. Paper \cite{Abd} discusses applications of the hierarchical conceptual model and the artificial neural networks-quantitative model to crude oil prices. In \cite{Zhao}, denoising autoencoders and bootstrap aggregation are combined to forecast crude oil prices. Paper  \cite{He} evaluates the accuracy of machine learning support vector regression models for forecasting crude oil prices. 

The application of machine learning to other financial data is also becoming more common. In \cite{Jiang}, a machine learning algorithm is applied to state-contingent claims and stochastic discount factors in financial markets. In \cite{Kul}, a machine learning algorithm is implemented to determine whether bank-differentiating factors influence firm choices in initial public offerings. In  \cite{Pas}, a multicriteria decision aid model is used in an attempt to replicate the credit ratings of Asian banks.

A commonly used stochastic model for the derivative and commodity market analysis is the Barndorff-Nielsen and Shephard (BN-S) model (see see \cite{BN1, BN-S1, BN-S2, BJS, Issaka, Issaka1, SWW, SWW2}). Though this model is very efficient and simple to use, it suffers from the absence of a long range dependence and many other issues. In this paper, we propose a simple way of improving the BN-S model with the implementation of various machine learning algorithms. After that, we validate the performance of the model. We use staging data sets that are close to production and see how our model behaves; if it gives good results, then the model is deployed and it is implemented. Finally, feedback is used to determine whether the model meets the business need for which it was built.

In this paper, we apply machine learning to the analysis of crude oil price data. In order to understand the data, we collect ten years of daily historical price data for crude oil.
After that, we conduct the exploratory data analysis. In the exploratory data analysis, we look at the basic statistics of the data such as its mean, median, and mode and correlations between the different labels. This exploratory data analysis gives direction to the model building.  Empirical analysis shows the presence of long memory in crude oil time series.  However, the intensity of the long-range dependence decreases over time. It is well  established that the classical BN-S model is not good for such data. In this paper, based on machine learning algorithms, we derive and implement a refined BN-S model to the crude oil price dynamics.

The organization of the paper is as follows. In Section \ref{sec2}, we briefly describe the BN-S model and why an improvement of this model is necessary for the analysis of crude oil price data. We find that the improvement of the model depends on machine learning analysis of the crude oil price data. The data analysis is provided in Section \ref{sec3}. A brief conclusion is provided in Section \ref{sec4}.

\section{An improved Barndorff-Nielsen and Shephard  model}
\label{sec2}

Many models in recent literature try to capture the stochastic behavior of time series. For example, in the case of the BN-S model, the stock or commodity price $S= (S_t)_{t \geq 0}$ on some filtered probability space $(\Omega, \mathcal{F}, (\mathcal{F}_t)_{0 \leq t \leq T}, \mathbb{P})$ is modeled by

\begin{equation}
\label{1}
S_t= S_0 \exp (X_t),
\end{equation}
\begin{equation}
\label{2}
dX_t = (\mu + \beta \sigma_t ^2 )\,dt + \sigma_t\, dW_t + \rho \,dZ_{\lambda t}, 
\end{equation}
\begin{equation}
\label{3}
d\sigma_t ^2 = -\lambda \sigma_t^2 \,dt + dZ_{\lambda t}, \quad \sigma_0^2 >0,
\end{equation}
where the parameters $\mu, \beta, \rho, \lambda \in \mathbb{R}$ with $\lambda >0$ and $\rho \leq 0$ and $r$ is the risk-free interest rate where a stock or commodity is traded up to a fixed horizon date $T$.  In this model $W_t$ is a Brownian motion and the process $Z_{t}$ is a subordinator. Also $W_t$ and $Z_t$ are assumed to be independent and $(\mathcal{F}_t)$ is assumed to be the usual augmentation of the filtration generated by the pair $(W_t, Z_t)$. 

However, the empirical data suggest that volatility ($\sigma_t$) usually fails to respond immediately to the sudden fluctuation of a stock or commodity price. The issue of the market's delayed response was raised in several papers (see \cite{BaT, Booth, Grinblatt}). Paper \cite{Arr} deals this issue with a delayed option price formula where the volatility has the form $\sigma(S_{t-b})$, for some delay parameter $b>0$.  

However, the results and the theoretical framework are far from satisfactory.  There are problems related to the above model:

\begin{enumerate}
\item Empirical results show that the jumps in volatility and stock or commodity price are positively correlated. However, unlike what is suggested by the model, they may not occur at the same time.  
\item For empirical data, the delay parameter $b$  is not deterministic. 
\item The performance of the model varies considerably depending both on the length and the density of time in the observed time series. Slow convergence is essentially caused by high serial correlation between the latent variables and the parameters. The problem is particularly acute in the case of a sparsely observed time series, or any case in which the time series contains many data.
\item The BN-S model does not incorporate the \emph{long range dependence} property. The model fails significantly for a longer range of time. In some occasions,  even for time spans as small as  two weeks, the model is unable to consistently capture the essential features of the related time series. 
\end{enumerate}

Some of these problems are addressed in various recent works. For example, in \cite{ijtaf}, the author presents a generalized version of the BN-S model. Assuming 
$Z_{t}$ and $Z_{t}^*$ to be two \emph{independent} L\'evy subordinators, define
\begin{equation}
\label{0}
d\tilde{Z}_{\lambda t}= \rho' \, dZ_{\lambda t} + \sqrt{1-\rho'^2}\, dZ_{\lambda t}^*,
\end{equation}
which is also a L\'evy subordinator provided $0 \leq \rho' \leq 1$. Thus, for $0 \leq \rho' \leq 1$, $Z$ and $\tilde{Z}$ are positively correlated L\'evy subordinators. Suppose the dynamics of $S_t$ are given by (\ref{1}), (\ref{2}), where $\sigma_t$ is given by
\begin{equation}
\label{4}
d\sigma_t ^2 = -\lambda \sigma_t^2 \,dt + d\tilde{Z}_{\lambda t}, \quad \sigma_0^2 >0,
\end{equation}
where $\tilde{Z}= (\tilde{Z}_{\lambda t})$ is a subordinator independent of $W$ but has a positive correlation with $Z$ as described above. Assume that the dynamics of $S= (S_t)$ is given by \eqref{1}, \eqref{2} and \eqref{4}. In \cite{ijtaf}, it is shown that this generalized model has the liberty to fit the option price and volatility in a correlated but different way, which is not possible for the case of the classical BN-S model.  This result is used for pricing vanilla options and developing theorems for parameter estimations of some particular variance processes.

The literature (see \cite{BN1, Semere}) shows that superpositions of Ornstein-Uhlenbeck (OU) type processes can be used to achieve \emph{long range dependence}. A limiting procedure creates processes that are self-similar with stationary increments. However, paper \cite{BN-S1} warns against fitting a large quantity of OU processes via a formal likelihood-based method. An alternative approach is to use heavy-tailed jump distributions in the model.

In this paper, we will address issues \#2, \#3 , and \#4 described above.  We will show that for crude oil price dynamics, the jump is \emph{not} completely stochastic. On the contrary, there is a \emph{deterministic} element in crude oil price that can be implemented to apply the existing models for an extended period of time. We will show from an empirical analysis that the dynamics of $X_t$ in \eqref{2} can be more accurately written when we use a convex combination of two independent subordinators, $Z$ and $Z^{(b)}$ as:

\begin{equation}
\label{2new}
dX_t = (\mu + \beta \sigma_t ^2 )\,dt + \sigma_t\, dW_t + \rho\left( (1-\theta) \,dZ_{\lambda t}+ \theta dZ^{(b)}_{\lambda t}\right), 
\end{equation}
where $\theta \in [0,1]$ is a \emph{deterministic} parameter. We will use several machine learning algorithms to determine the value of $\theta$. The process $Z^{(b)}$ in \eqref{2new} is a subordinator that has greater intensity than the subordinator $Z$. In \eqref{2new}, $\lambda>0$ is a scale parameter for the time. The subordinator $Z^{(b)}$, that has greater intensity than $Z$, corresponds to a greater L\'evy density subordinator. For instance, if the L\'evy densities of $Z$ and $Z^{(b)}$ are given by $\nu_1 \alpha e^{-\alpha x}$ and $\nu_2 \alpha e^{-\alpha x}$, respectively (for $\alpha, \nu_1, \nu_2>0$, and $x>0$), then $\nu_2>\nu_1$.

In this case \eqref{4} will be given by
\begin{equation}
\label{4new} 
d\sigma_t ^2 = -\lambda \sigma_t^2 \,dt + (1- \theta') dZ_{\lambda t} + \theta' dZ^{(b)}_{\lambda t} , \quad \sigma_0^2 >0,
\end{equation}
where, as before, $\theta' \in [0,1]$ is \emph{deterministic}. For simplicity, we assume $\theta= \theta'$ for the rest of this paper.

\begin{theorem}
\label{big11}
If the jump measure associated with the subordinator $Z$ be $J_Z$, and $J(s)= \int_0^s \int_{\mathbb{R}^+} J_Z(\lambda d\tau, dy)$, then for the log-return of the classical BN-S model given by \eqref{1}, \eqref{2}, and \eqref{3},
\begin{equation}
\label{corrBNS}
\text{Corr}(X_t, X_s)= \frac{\int_0^s \sigma_{\tau}^2 d\tau + \rho^2 J(s)}{ \sqrt{\left(\int_0^t \sigma_{\tau}^2 d\tau + t\rho^2 \lambda \text{Var}(Z_1)\right) \left(\int_0^s \sigma_{\tau}^2 d\tau + s\rho^2 \lambda \text{Var}(Z_1)\right)}},
\end{equation}
for $t>s$. 
\end{theorem}
\begin{proof}
Clearly, for $t>s$,
$$\text{Cov}(X_t, X_s)= \int_0^s \sigma_{\tau}^2 d\tau + \rho^2 \int_0^s \int_{\mathbb{R}^+} J_Z(\lambda d\tau, dy).$$
Note that the instantaneous variance of the log-return is given by $(\sigma_t^2 + \rho^2 \lambda \text{Var}(Z_1))\,dt$. Consequently we obtain \eqref{corrBNS}. 
\end{proof}

Note that for a fixed $s$, if $t$ increases, then $\text{Corr}(X_t, X_s)$ quickly decreases. The proof of the following result is very similar to the proof of Theorem \ref{big11}.

\begin{theorem}
\label{big12}
If the jump measures associated with the subordinators $Z$ and $Z^{(b)}$ are $J_Z$ and $J^{(b)}_Z$ respectively, and $J(s)= \int_0^s \int_{\mathbb{R}^+} J_Z(\lambda d\tau, dy)$, $J^{(b)}(s)= \int_0^s \int_{\mathbb{R}^+} J^{(b)}_Z(\lambda d\tau, dy)$; then for the log-return of the refined BN-S model given by \eqref{1}, \eqref{2new}, and \eqref{4new},
\begin{align}
\label{corrBNSimproved}
\text{Corr}(X_t, X_s)= \frac{\int_0^s \sigma_{\tau}^2 d\tau + \rho^2  (1-\theta)^2 J(s) + \rho^2 \theta^2  J^{(b)}(s)}{ \sqrt{\alpha(t) \alpha(s)}},
\end{align}
for $t>s$, where
$\alpha(\nu) = \int_0^{\nu} \sigma_{\tau}^2 d\tau + \nu\rho^2 \lambda ( (1-\theta)^2 \text{Var}(Z_1)+  \theta^2 \text{Var}(Z^{(b)}_1)) $.
\end{theorem}

\begin{proof}
We observe for $t>s$,
$$\text{Cov}(X_t, X_s)= \int_0^s \sigma_{\tau}^2 d\tau + \rho^2  (1-\theta)^2 \int_0^s \int_{\mathbb{R}^+} J_Z(\lambda d\tau, dy) + \rho^2 \theta^2  \int_0^s \int_{\mathbb{R}^+} J^{(b)}_Z(\lambda d\tau, dy).$$
Also, the variance of the log-returns $X_t$ and $X_s$ are given by $\int_0^{t} \sigma_{\tau}^2 d\tau + \nu\rho^2 \lambda ( (1-\theta)^2 \text{Var}(Z_1)+  \theta^2 \text{Var}(Z^{(b)}_1))$, and $\int_0^{s} \sigma_{\tau}^2 d\tau + \nu\rho^2 \lambda ( (1-\theta)^2 \text{Var}(Z_1)+  \theta^2 \text{Var}(Z^{(b)}_1)) $, respectively. Consequently we obtain \eqref{corrBNSimproved}. 
\end{proof}

Note that as $\theta$ is constantly adjusted, for a fixed $s$, the value of $t$ always has an upper limit. Consequently, $\text{Corr}(X_t, X_s)$ never becomes ``too small". This is the major difference between the results in Theorem \ref{big11} and Theorem \ref{big12}.

The advantages of the dynamics given by the refined BN-S model given by \eqref{1}, \eqref{2new}, and \eqref{4new}, over the existing models are significant. First of all, this minor change in the model incorporates \emph{long range dependence} without actually changing the model. This model will be more efficient, but at the same time have many fewer parameters than the \emph{superposition} models.  Secondly, the performance of this model for a sparsely observed time series will be improved. Thirdly, an estimation the delay parameter $b$ (mentioned in \#2) can be obtained. Finally, and possibly most importantly, the procedure and the model show the application of data science for extracting a \emph{deterministic component} out of processes that are thus far considered to be completely stochastic. For this paper, we restrict our analysis for crude oil price dynamics. However, this method possibly can be implemented for any compatible time series.

\section{Data analysis}
\label{sec3}

A commodity of fundamental importance is the crude oil. Consequently an analysis of the dynamics of crude oil price time series seems to be crucial. This allows to ascertain the potential impacts of its shocks in several economies and on other financial assets (see \cite{Tabak}). As observed in \cite{Sensoy}, long-range dependence is evident in various energy futures markets.  Empirical analysis shows the presence of long memory in crude oil time series.  However, the intensity of the long-range dependence decreases over time. As described in the beginning of Section \ref{sec2}, the classical BN-S model is not good for such data. On the other hand, Theorem \ref{big12} shows that the refined BN-S model proposed in this paper can be implemented in this case. 

We consider crude oil price data over a period of 10 years. We use the West Texas Intermediate (WTI or NYMEX) crude oil prices data set for the period June 1, 2009 to May 30, 2019 (Figure 1). There are a total of $2,530$ available data in this set. For convenience, we index the dates (for available data)  from 0 (for June 1, 2009) to 2529 (for  May 30, 2019). The following table (Table 1) summarizes various estimates for the data set.

\begin{table}[H]
\centering
\caption{Properties of the empirical data set.}
  \begin{tabular}{ | l | c | r |}
    \hline
    & Daily Price Change & Daily Price Change \% \\ \hline
    Mean & -0.0047 & 0.01370 \% \\ \hline
    Median & 0.04399 &  0.06521 \%\\     \hline
Maximum & 7.62 &  12.32 \%\\     \hline
Minimum & -8.90 &  -10.53 \%\\     \hline
  \end{tabular}
\end{table}

\begin{figure}[H]
\centering
\caption{Crude oil close price.}
\includegraphics[scale=.6]{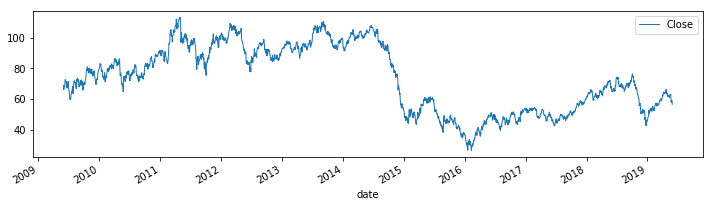}
\end{figure}

We implement the following \emph{procedure} (Step 1 through Step 5) that creates a \emph{classification problem} for the data set. For the data set:
\begin{enumerate}
\item We conduct exploratory data analysis. 
\item We consider the daily \emph{close price} for the historical oil price data. From the plots we identify a value of $K$ to define a ``big jump" in the crude oil close price. We identify the dates for which the close price is $K$ ``points" less than the close price of the previous day (for example, if $K=1\%$, we will find the dates for which the close price is $1\%$ below the previous business day). 
\item We create a new data-frame from the old one where ``features" (columns) will be seven consecutive close prices. For example, if the close prices are $$a_1,a_2,a_3,a_4,a_5,a_6,a_7,a_8,a_9,a_{10},\cdots;$$ then the first row of the data set will contain $$a_1,a_2,a_3,a_4,a_5,a_6,a_7;$$ second row of the data set will contain $$a_2,a_3,a_4,a_5,a_6,a_7,a_8;$$  etc. 
\item We create a new target column for the new data-frame (as created in the preceding step) as follows: $\theta=1$ for those set of seven close prices that immediately precede \emph{at least two jumps} of size $K$ (or more) in the following seven days. Otherwise we label the target column by $\theta=0$.

 For example: suppose we identified $a_{8}$ and $a_{10}$ as ``big jumps". Then the $\theta=1$  for the first row $a_1,a_2,a_3,a_4,a_5,a_6,a_7$.
\item We run various \emph{classification algorithms} from machine learning where the input is the \emph{close price for seven consecutive days}, and output is $\theta$-value (0 or 1). 
We evaluate the classification report and confusion matrix in each case. 
\end{enumerate}

We will show that we can find $\theta$ with reasonable accuracy and use this for \eqref{2new}. The result can be improved by adjusting the value of $K$ in Step 2. The result can be further improved by increasing the number of days (in Step 3) from seven to a higher number. It is worth noting that the various deep learning models provide a value of $\theta$ between $0$ and $1$. In Step 4, we approximate that by $0$ or $1$. However, the actual value of $\theta$ may be directly used in \eqref{2new}.

Figures 2, 3, and 4 provide various visualizations of crude oil close prices. Figures 5 and 6 provide a histogram of the daily price change and  a histogram of daily percentage change, respectively. We partition this data set in various ways. For each partition we use a \emph{train-test-split}, with respect to a given date.  We summarize the list of figures.

\begin{itemize}
\item[Figure 1:] West Texas Intermediate (WTI or NYMEX) crude oil prices data set for the period June 1, 2009 to May 30, 2019 (crude oil close price).
\item[Figure 2:] Yearly boxplot for the close oil price.
\item[Figure 3:] Distribution plot for close oil price.
\item[Figure 4:] Bar chart for close oil price.
\item[Figure 5:] Histogram for daily change in close oil price.
\item[Figure 6:] Histogram for daily change percentage in close oil price.
\end{itemize}

\begin{figure}[H]
\centering
\caption{Yearly boxplot for the close oil price.}
\includegraphics[scale=.5]{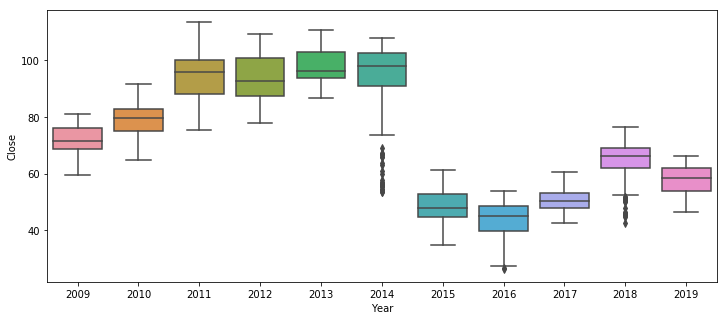}
\end{figure}
\begin{figure}[H]
\centering
\caption{Distribution plot for close oil price.}
\includegraphics[scale=.5]{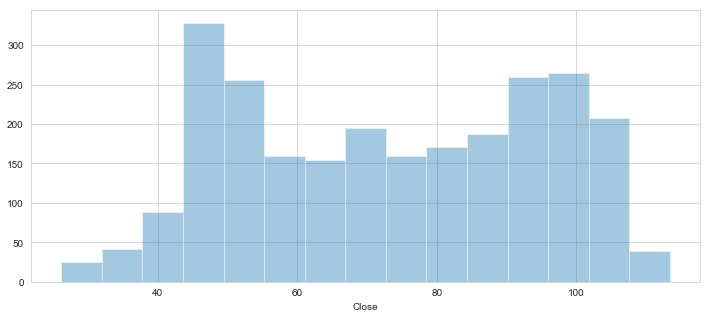}
\end{figure}
\begin{figure}[H]
\centering
\caption{Bar chart for close oil price.}
\includegraphics[scale=.9]{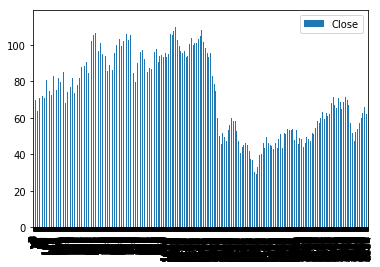}
\end{figure}

\begin{figure}[H]
\centering
\caption{Histogram for daily change in close oil price.}
\includegraphics[scale=.9]{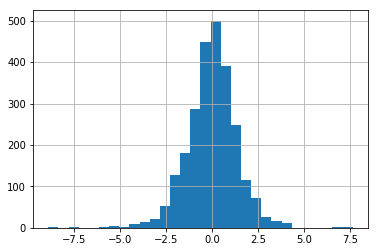}
\end{figure}
\begin{figure}[H]
\centering
\caption{Histogram for daily change percentage in close oil price.}
\includegraphics[scale=.9]{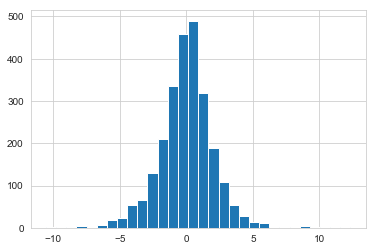}
\end{figure}

For the following analysis we use $K= 2\%$, i.e., $\theta=1$ for the set of seven close prices that immediately precede \emph{at least two jumps} of size $2\%$ (or more) in the following seven days. Otherwise, we use $\theta=0$. 

We run various \emph{supervised learning algorithms} on the crude oil price data. We begin with the \emph{logistic regression} (LR) and the \emph{random forest} (RF) classification of the data set. It is well known that for the \emph{logistic regression} classification, given a testing data $X$, $\mathbb{P}(\theta =1| X) = \frac{1}{1+ e^{-\beta_0 - \beta_1 \cdot X}}$, where the quantity $\beta_0$ and the vector $\beta_1$ are determined from the training set with the help of an appropriate log-likelihood function. The \emph{random forest} classification of many decision trees with a random sample of features is used. By randomly leaving out candidate features from each split, \emph{random forest} \emph{decorrelates} the trees, such that the averaging process can reduce the variance of the resulting model.

After that, we implement various \emph{deep learning} techniques:

\begin{enumerate}
\item[(A)] A \emph{neural network with two hidden layers} (with activations $\tanh$ and $\text{ReLU}$) and an output layer (with a $\text{softmax}$ activation function). For simplicity we approximate $\theta$ in \eqref{2new} with 0 (``no big jump") and 1 (``big jump"). For this approximation, we take $\theta=1$ if the output probability for the $\text{softmax}$ activation function corresponding to $\theta=1$ is more than $0.3$.
\item[(B)] \emph{Long short-term memory} (LSTM) along with the neural network described in (A).  LSTM is an artificial recurrent neural network (RNN) architecture that is implemented to avoid the \emph{vanishing gradient problem}. The \emph{vanishing gradient problem} is especially prominent when a vanilla RNN, constructed from regular neural network nodes, is implemented to model dependencies between time series values that are separated by a significant number of days. LSTM has in-built feedback connections that make it appropriately implementable for a financial time series. A common LSTM unit is composed of a cell, an input gate, an output gate, and a forget gate. The cell retains values over arbitrary time intervals and the other three gates regulate the flow of information into and out of the cell.
\item[(C)] \emph{LSTM along with a batch normalizer} (BN) and the neural network described in (A). A \emph{batch normalizer} standardizes and rescales the output of a given layer in the deep network.  To increase the stability of a neural network, batch normalization normalizes the output of a previous activation layer by subtracting the batch mean and dividing by the batch standard deviation. It also reduces the amount the hidden unit values shift around (i.e., its \emph{covariance shift}). This process centers all the inputs around zero. This way, there is not much change in each layer input. Consequently, layers in the network can learn from the back-propagation simultaneously, without waiting for the previous layer to learn. This speeds up the training of networks.
\end{enumerate}

Note that, once the value of $\theta$ is obtained from the training data, we use this value for the refined BN-S model (in \eqref{2new} and \eqref{4new}). In particular, we use this \emph{deterministic} $\theta$ value for the testing data. In addition to that, this \emph{deterministic} value of $\theta$ can be used for prediction using the refined BN-S model. 

For the following tables (Table 2 through Table 13), we provide classification reports for various machine learning algorithms. For the testing data, \emph{true positive}, \emph{true negative}, \emph{false positive}, and \emph{false negative} are denoted as TP, TN, FP, and FN, respectively. In the context of this study, ``TP" and ``TN" are the cases when the model correctly predicts $\theta=1$, and $\theta=0$, respectively. Also, in the context of this study, ``FP" is the case when $\theta=0$ is predicted as $\theta=1$; and ``FN" is the case when $\theta=1$ is predicted as $\theta=0$. The following measurements are standard:

$$\textit{precision}= \frac{\text{TP}}{\text{TP}+ \text{FP}},$$

$$\textit{recall}= \frac{\text{TP}}{\text{TP}+ \text{FN}}.$$
The \emph{f1-score} gives the harmonic mean of precision and recall. The scores corresponding to every class gives the accuracy of the classifier in classifying the data points in that particular class compared to all other classes. The \emph{support} is the number of samples of the true response that lie in that class.

\begin{table}[H]
\caption{Various estimations for \emph{training date(index)}: October 21, 2009 (100) to May 24, 2011 (500); and \emph{testing date(index)}: May 25, 2011 (501) to October 14, 2011 (600). }
  \begin{tabular}{ | c | c | c |c|c|c|}
   \hline
    & LR & RF & Neural Network (A) & LSTM (B) & BN (C) \\ \hline
    precision $\theta=0$ & 0.56 &  0.57 & 0.56 & 0.62 &  0.56 \\ \hline
    recall $\theta=0$ & 0.96 & 0.91  & 0.70 & 0.79  & 0.65 \\     \hline
f1-score $\theta=0$ & 0.71 & 0.70  & 0.62 &0.69  & 0.60 \\     \hline
support $\theta=0$ & 57 & 57  & 57 & 57 & 57 \\     \hline
    precision $\theta=1$ & 0.00 & 0.50 & 0.41 & 0.57& 0.43 \\ \hline
    recall $\theta=1$ & 0.00 & 0.11  & 0.27 & 0.36 & 0.34 \\     \hline
f1-score $\theta=1$ & 0.00 & 0.19  & 0.33 & 0.44 & 0.38 \\     \hline
support $\theta=1$ & 44 & 44 & 44 & 44& 44 \\     \hline
\end{tabular}
\end{table}

\begin{table}[H]
\caption{Various estimations for \emph{training date(index)}: : October 21, 2009 (100) to October 14, 2011 (600); and \emph{testing date(index)}: October 17, 2011 (601) to August 1, 2012 (800). }
  \begin{tabular}{ | c | c | c |c|c|c|}
   \hline
    & LR & RF & Neural Network (A) & LSTM (B) & BN (C) \\ \hline
    precision $\theta=0$ &  0.83 &  0.83 & 0.83 & 0.83 &  0.81\\ \hline
    recall $\theta=0$ &  0.99 & 0.91  &0.54 & 0.62 & 0.62 \\     \hline
f1-score $\theta=0$ & 0.91 & 0.87  & 0.65 & 0.71 & 0.70 \\     \hline
support $\theta=0$ & 168 & 168  &168 & 168 & 168 \\     \hline
    precision $\theta=1$ & 0.00  & 0.12 & 0.15& 0.16 &  0.11\\ \hline
    recall $\theta=1$ & 0.00 & 0.06  & 0.42& 0.36 & 0.24 \\     \hline
f1-score $\theta=1$ & 0.00 & 0.08  & 0.23& 0.22 & 0.15 \\     \hline
support $\theta=1$ & 33 & 33 & 33& 33 & 33\\     \hline
\end{tabular}
\end{table}

\begin{table}[H]
\caption{Various estimations for \emph{training  date(index)}: August 9, 2010 (300) to August 1, 2012 (800); and \emph{testing date(index)}: August 2, 2012 (801) to May 17, 2013 (1000). }
  \begin{tabular}{ | c | c | c |c|c|c|}
   \hline
    & LR & RF & Neural Network (A) & LSTM (B) & BN (C) \\ \hline
    precision $\theta=0$ &  0.92 &  0.92 & 0.91 & 0.91 &  0.92 \\ \hline
    recall $\theta=0$ &  1.00 & 0.92  & 0.58 & 0.58 & 0.58 \\     \hline
f1-score $\theta=0$ & 0.96 & 0.92  & 0.71 & 0.71  & 0.71 \\     \hline
support $\theta=0$ &  185& 185  & 185 & 185 & 185\\     \hline
    precision $\theta=1$ & 0.00 &  0.07 &  0.06 & 0.07 &  0.07 \\ \hline
    recall $\theta=1$ & 0.00  & 0.06  & 0.31 & 0.38 & 0.38 \\     \hline
f1-score $\theta=1$ & 0.00 & 0.06  & 0.10 & 0.12 & 0.12 \\     \hline
support $\theta=1$ & 16 & 16 & 16& 16 & 16\\     \hline
\end{tabular}
\end{table}

\begin{table}[H]
\caption{Various estimations for \emph{training  date(index)}: May 17, 2013 (1000) to December 17, 2014 (1400); and \emph{testing date(index)}: December 18, 2014 (1401) to May 13, 2015 (1500). }
  \begin{tabular}{ | c | c | c |c|c|c|}
   \hline
    & LR & RF & Neural Network (A) & LSTM (B) & BN (C) \\ \hline
    precision $\theta=0$ &  0.48 & 0.48 & 0.47 & 0.48 &  0.50 \\ \hline
    recall $\theta=0$ &  1.00 & 1.00  & 0.96 &1.00  & 0.98 \\     \hline
f1-score $\theta=0$ &0.64 & 0.65  & 0.63 &0.65 & 0.66 \\     \hline
support $\theta=0$ & 48 & 48  & 48 & 48 & 48 \\     \hline
    precision $\theta=1$ & 0.00 &  1.00&  0.50 & 1.00 & 0.86  \\ \hline
    recall $\theta=1$ & 0.00 &  0.02 & 0.04 & 0.04 & 0.11 \\     \hline
f1-score $\theta=1$ & 0.00  &  0.04 & 0.07 & 0.07 & 0.20 \\     \hline
support $\theta=1$ & 53 & 53  & 53 & 53 & 53\\     \hline
\end{tabular}
\end{table}

\begin{table}[H]
\caption{Various estimations for \emph{training  date(index)}: March 5, 2014 (1200) to May 13, 2015 (1500); and \emph{testing date(index)}: May 14, 2015 (1501) to October 5, 2015 (1600). }
  \begin{tabular}{ | c | c | c |c|c|c|}
   \hline
    & LR & RF & Neural Network (A) & LSTM (B) & BN (C) \\ \hline
    precision $\theta=0$ &  0.45 &  0.48 & 0.52 & 0.43 & 0.48 \\ \hline
    recall $\theta=0$ &  0.96 & 0.94  & 0.62 & 0.45 &0.83 \\     \hline
f1-score $\theta=0$ & 0.62 & 0.64  & 0.56  &0.44 & 0.61 \\     \hline
support $\theta=0$ & 47 & 47   & 47 & 47 & 47\\     \hline
    precision $\theta=1$ & 0.00 & 0.70 & 0.60 & 0.50 &  0.60 \\ \hline
    recall $\theta=1$ & 0.00  &  0.13 & 0.50 & 0.48 & 0.22 \\     \hline
f1-score $\theta=1$ & 0.00  &  0.22 & 0.55 & 0.49 & 0.32\\     \hline
support $\theta=1$ & 54 & 54  & 54 & 54 & 54\\     \hline
\end{tabular}
\end{table}

\begin{table}[H]
\caption{Various estimations for \emph{training  date(index)}: July 28, 2014 (1300) to February 29, 2016 (1700); and \emph{testing  date(index)}: March 1, 2016 (1701) to November 29, 2016 (1900). }
  \begin{tabular}{ | c | c | c |c|c|c|}
   \hline
    & LR & RF & Neural Network (A) & LSTM (B) & BN (C) \\ \hline
    precision $\theta=0$ &0.56  & 0.51 & 0.54 & 0.61 &  0.50 \\ \hline
    recall $\theta=0$ &  0.17 & 0.59  & 0.12 & 0.10 & 0.06 \\     \hline
f1-score $\theta=0$ & 0.26 &  0.55 & 0.20 & 0.17 & 0.11\\     \hline
support $\theta=0$ & 114 & 114  & 114 & 114 & 114\\     \hline
    precision $\theta=1$ &  0.43 & 0.33 &  0.43 & 0.44 & 0.43  \\ \hline
    recall $\theta=1$ & 0.83 &  0.26 &  0.86 & 0.92& 0.92\\     \hline
f1-score $\theta=1$ & 0.57 &  0.29 & 0.57 & 0.59 & 0.58\\     \hline
support $\theta=1$ & 87 & 87  & 87 & 87 & 87\\     \hline
\end{tabular}
\end{table}

\begin{table}[H]
\caption{Various estimations for \emph{training  date(index)}: July 28, 2014 (1300) to July 12, 2016 (1800); and \emph{testing  date(index)}: July 13, 2016 (1801) to April 21, 2017 (2000). }
  \begin{tabular}{ | c | c | c |c|c|c|}
   \hline
    & LR & RF & Neural Network (A) & LSTM (B) & BN (C) \\ \hline
    precision $\theta=0$ & 0.64  &  0.66& 0.73 & 0.69&  0.69 \\ \hline
    recall $\theta=0$ & 0.60  & 0.71  & 0.26  & 0.25 & 0.18 \\     \hline
f1-score $\theta=0$ & 0.62 & 0.68 & 0.38 & 0.37 & 0.29 \\     \hline
support $\theta=0$ & 136  & 136   & 136 & 136 & 136\\     \hline
    precision $\theta=1$ & 0.26  & 0.26  & 0.34 & 0.33 & 0.33 \\ \hline
    recall $\theta=1$ & 0.29 &  0.22 & 0.80 & 0.77& 0.83 \\     \hline
f1-score $\theta=1$ & 0.28 &  0.24 & 0.48 &0.46 & 0.47\\     \hline
support $\theta=1$ & 65 & 65  & 65 & 65 & 65\\     \hline
\end{tabular}
\end{table}

\begin{table}[H]
\caption{Various estimations for \emph{training date(index)}: May 13, 2015 (1500) to April 21, 2017 (2000); and \emph{testing  date(index)}: April 24, 2017 (2001) to September 8, 2017 (2100). }
  \begin{tabular}{ | c | c | c |c|c|c|}
   \hline
    & LR & RF & Neural Network (A) & LSTM (B) & BN (C) \\ \hline
    precision $\theta=0$ & 0.75  &  0.77 & 0.65 & 0.81 &  0.72 \\ \hline
    recall $\theta=0$ & 1.00  & 0.82  & 0.22 & 0.38  & 0.34\\     \hline
f1-score $\theta=0$ & 0.86 & 0.79  & 0.33  & 0.52 & 0.46\\     \hline
support $\theta=0$ & 76 & 76   & 76 & 76 & 76 \\     \hline
    precision $\theta=1$ &  0.00 &  0.30 & 0.21  & 0.28& 0.23  \\ \hline
    recall $\theta=1$ & 0.00 & 0.24  &  0.64 & 0.72 & 0.60\\     \hline
f1-score $\theta=1$ & 0.00 & 0.27  & 0.32 & 0.40 & 0.33\\     \hline
support $\theta=1$ & 25 & 25  & 25 & 25 & 25\\     \hline
\end{tabular}
\end{table}

\begin{table}[H]
\caption{Various estimations for \emph{training date(index)}: October 5, 2015 (1600) to September 8, 2017 (2100); and \emph{testing date(index)}: September 11, 2017 (2101) to February 1, 2018 (2200). }
  \begin{tabular}{ | c | c | c |c|c|c|}
   \hline
    & LR & RF & Neural Network (A) & LSTM (B) & BN (C) \\ \hline
    precision $\theta=0$ &  0.92  &  0.93& 0.95 & 0.92 &  0.93 \\ \hline
    recall $\theta=0$ &  1.00 &  0.96 &0.39 & 0.76 & 0.67 \\     \hline
f1-score $\theta=0$ & 0.96 & 0.94 & 0.55 & 0.84 & 0.78\\     \hline
support $\theta=0$ & 93 & 93   & 93 & 93 & 93\\     \hline
    precision $\theta=1$ &  0.00&  0.20& 0.10 & 0.08 & 0.09 \\ \hline
    recall $\theta=1$ & 0.00 &  0.12 & 0.75 & 0.25 & 0.38 \\     \hline
f1-score $\theta=1$ & 0.00 &  0.15 & 0.17 & 0.12 & 0.14\\     \hline
support $\theta=1$ & 8  & 8 & 8 & 8 & 8\\     \hline
\end{tabular}
\end{table}

\begin{table}[H]
\caption{Various estimations for \emph{training date(index)}: February 29, 2016 (1700) to February 1, 2018 (2200); and \emph{testing date(index)}: February 2, 2018 (2201) to June 26, 2018 (2300). }
  \begin{tabular}{ | c | c | c |c|c|c|}
   \hline
    & LR & RF & Neural Network (A) & LSTM (B) & BN (C) \\ \hline
    precision $\theta=0$ &  0.94 &  0.93& 0.96 & 0.95&  0.94 \\ \hline
    recall $\theta=0$ & 1.00  &  0.84 &0.67 & 0.56 & 0.67\\     \hline
f1-score $\theta=0$ & 0.97 & 0.88 & 0.79 & 0.70 & 0.79\\     \hline
support $\theta=0$ &  95 & 95  & 95 & 95 & 95\\     \hline
    precision $\theta=1$ & 0.00 &  0.00 & 0.09 & 0.07&  0.06\\ \hline
    recall $\theta=1$ & 0.00 &  0.00 & 0.50 & 0.50 & 0.33\\     \hline
f1-score $\theta=1$ & 0.00 & 0.00  & 0.15 & 0.12 & 0.10\\     \hline
support $\theta=1$ & 6 & 6 & 6& 6&6\\     \hline
\end{tabular}
\end{table}

\begin{table}[H]
\caption{Various estimations for \emph{training date(index)}: July 12, 2016 (1800) to June 26, 2018 (2300); and \emph{testing date(index)}: June 27, 2018 (2301) to November 14, 2018 (2400). }
  \begin{tabular}{ | c | c | c |c|c|c|}
   \hline
    & LR & RF & Neural Network (A) & LSTM (B) & BN (C) \\ \hline
    precision $\theta=0$ &  0.74  & 0.76 & 0.75 & 0.76& 0.78 \\ \hline
    recall $\theta=0$ &  1.00  &  0.99 & 0.99 & 0.87 & 0.79 \\     \hline
f1-score $\theta=0$ & 0.85 & 0.86 & 0.85 & 0.81 & 0.78 \\     \hline
support $\theta=0$ &  75 & 75  & 75 & 75 & 75\\     \hline
    precision $\theta=1$ &  0.00 &  0.67& 0.50 & 0.38 & 0.36 \\ \hline
    recall $\theta=1$ & 0.00 &  0.08 & 0.04 & 0.23 & 0.35 \\     \hline
f1-score $\theta=1$ &  0.00 & 0.14  & 0.07 & 0.29 & 0.35 \\     \hline
support $\theta=1$ &  26 & 26  & 26 & 26 & 26\\     \hline
\end{tabular}
\end{table}

\begin{table}[H]
\caption{Various estimations for \emph{training date(index)}: July 12, 2016 (1800) to  June 26, 2018 (2300); and \emph{testing date(index)}: June 27, 2018 (2301) to April 10, 2019 (2500). }
  \begin{tabular}{ | c | c | c |c|c|c|}
   \hline
    & LR & RF & Neural Network (A) & LSTM (B) & BN (C) \\ \hline
    precision, $\theta=0$ & 0.77  & 0.78  &  0.77& 0.79 &  0.83 \\ \hline
    recall $\theta=0$ & 1.00  & 0.96  & 0.92 & 0.92 & 0.75 \\     \hline
f1-score $\theta=0$ &0.87  & 0.86 & 0.84 & 0.85 & 0.79\\     \hline
support $\theta=0$ & 154 &  154 & 154 & 154 & 154\\     \hline
    precision, $\theta=1$ &  0.00 & 0.45 & 0.32 & 0.43 &  0.38 \\ \hline
    recall $\theta=1$ & 0.00 &  0.11 & 0.13 & 0.21 & 0.49 \\     \hline
f1-score $\theta=1$ & 0.00 & 0.17  & 0.18 & 0.29 & 0.43 \\     \hline
support $\theta=1$ & 47 & 47  & 47 & 47 & 47\\     \hline
\end{tabular}
\end{table}

To make the BN-S model implementable for a long range, it is clear that a single L\'evy subordinator is not effective. If a large fluctuation in the future can be apprehended from the historical data (i.e., $\theta=1$) with the help of machine learning algorithms, we can ``switch" the initial L\'evy subordinator ($Z$) to the more intense L\'evy subordinator ($Z^{(b)}$) that corresponds to larger fluctuations. On the other hand if no big fluctuation in the future can be apprehended from the historical data (i.e., $\theta=0$) with the help of machine learning algorithms, we can ``switch" the L\'evy subordinator $Z^{(b)}$ to $Z$. In this way, a single equation \eqref{2new} can be used to describe the crude oil dynamics even for a longer time period. 

It is clear from the various tables that the \emph{logistic regression} is less efficient in detecting future big jumps ($\theta =1$) based on the historical data. For most of the cases the neural network technique (A), LSTM (B), or the LSTM with a batch normalizer (C), work better than the random forest classifier. Also, if the algorithms are trained on more data points, the predictions for $\theta=1$ are better. To keep the model simple, only two hidden layers are used. The results improve if the number of hidden layers is increased. Also, note that the softmax activation function in the output layers for (A), (B), or (C), in fact provides probabilities for $\theta=0$ and $\theta=1$. With appropriate scaling those probabilities can be used in lieu of $(1-\theta)$ and $\theta$ in \eqref{2new}. 

Once we have a good estimation of the value of $\theta$, we can implement that to \eqref{2new}. That would lead to one of two options: (1) if the initial description of the BN-S dynamics incorporates $Z$ (or $Z^{(b)}$) as the L\'evy subordinator and $\theta=0$ is established, we continue (or, update) the subordinator with $Z$; (2)  if the initial description of the BN-S dynamics incorporates $Z$ (or $Z^{(b)}$) as the L\'evy subordinator and $\theta=1$ is established, we update (or, continue) the subordinator with $Z^{(b)}$.  The machine learning algorithms can be performed dynamically in order to continue or update with the background driving L\'evy process in the BN-S model. 

As a result, the analysis shows that for crude oil price dynamics, the jump is \emph{not} completely stochastic. There is a \emph{deterministic} element ($\theta$) in it that can be implemented to apply the existing models for an extended period of time. Thus the new model incorporates long term dependence without changing the tractability of the model. This model is more efficient, but at the same time has many fewer parameters than the \emph{superposition} models.  

\section{Conclusion}
\label{sec4}

We observe that a classical BN-S model may not appropriately represent crude oil price dynamics. In this paper, we implement various machine learning algorithms to determine the possibility of an upcoming large fluctuation in the crude oil price. Once those \emph{possibilities} are obtained, the classical BN-S model is modified (or not, depending on the obtained \emph{possibilities}) with respect to its background driving L\'evy subordinator. This modification enables \emph{long range dependence} in the new model without significantly changing the model. Also, this modification incorporates only one extra parameter (i.e., $\theta$)  compared to the classical model. It is shown in this paper that the parameter $\theta$ is \emph{deterministic} and can be obtained from the empirical data using various machine learning techniques. 

In this paper we implement machine learning algorithms to the empirical data in order to improve the mathematical model for commodity price dynamics. In a sequel of this work, we plan to implement this analysis for other financial time series. Also, we observe that the stochastic equation related to the volatility dynamics does not play a crucial role in the present analysis. The situation will be different and improved if it can be appropriately analyzed for an empirical data set. \\

\textbf{Acknowledgment}: The authors would like to thank the anonymous reviewers for their careful reading of the manuscript and for suggesting points to improve the quality of the paper.


\begin{thebibliography}{1}

\footnotesize

\bibitem{Abd} Abdullah S. N. \& Zeng X. (2010), Machine learning approach for crude oil price prediction with Artificial Neural Networks-Quantitative (ANN-Q) model, \emph{The 2010 International Joint Conference on Neural Networks (IJCNN)}, doi: 10.1109/IJCNN.2010.5596602

\bibitem{Arr} Arriojas M., Hu Y., Mohammed S-E. \& Pap G. (2007), A Delayed Black and Scholes Formula, \emph{Stoch Anal Appl.}, \textbf{25},  471–492.

\bibitem{BN1}  Barndorff-Nielsen O. E. (2001), Superposition of Ornstein-Uhlenbeck Type Processes, \emph{Theory Probab. Appl.}, \textbf{45}, 175-194.

\bibitem{BaT} Bernard V. \& Thomas J. (1989), Post-earnings-announcement drift: delayed price response or risk premium?, \emph{J. Account. Res.}, \textbf{27}, 1-36.

\bibitem{BN-S1} Barndorff-Nielsen O. E. \& Shephard N.(2001),  Non-Gaussian Ornstein-Uhlenbeck-based models and some of their uses in financial economics, \emph{J. R. Stat. Soc. Ser. B Stat. Methodol.}, \textbf{63}, 167-241.

\bibitem{BN-S2}  Barndorff-Nielsen O. E. \& Shephard N. (2001),  Modelling by L\'evy Processes for Financial Econometrics, In \emph{L\'evy Processes : Theory and Applications} (eds O. E. Barndorff-Nielsen, T. Mikosch \& S. Resnick), 283-318, Birkh\"auser.

\bibitem{BJS}  Barndorff-Nielsen O. E. , Jensen J. L. \& S$\o$rensen M. (1998), Some stationary processes in discrete and continuous time, \emph{Adv. in Appl. Probab.}, \textbf{30}, 989-1007.

\bibitem{BKR}  Benth F. E., Karlsen K. H. \& K. Reikvam (2003),  Merton's portfolio optimization problem in a Black and Scholes market with non-Gaussian stochastic volatility of Ornstein-Uhlenbeck type, \emph{Math. Finance}, \textbf{13}, 215-244.

\bibitem {BS}  Black F. \& Scholes M. (1973),  The pricing of options and corporate liabilities, \emph{J. Political Econ.}, \textbf{81}, 637-659.

\bibitem{Booth} Booth G., Kallunki J., \& Martikainen T. (1997), Delayed price response to the announcements of earnings and its components in Finland, \emph{European Account.
Rev.}, \textbf{6}, 377-392.

\bibitem{Sircar2} Brown I., Funk J., \&  Sircar R. (2017), Oil Prices \& Dynamic Games Under Stochastic Demand, Available at SSRN: \url{https://ssrn.com/abstract=3047390 or http://dx.doi.org/10.2139/ssrn.3047390}.

\bibitem{Sircar1} Chan P. \& Sircar R. (2017), Fracking, Renewables, and Mean Field Games, \emph{SIAM Review}, \textbf{59}(3), 588-615.

\bibitem{oil11} Chen Y., Kaijian H. \& Tso G. K.F. (2017), Forecasting Crude Oil Prices: a Deep Learning based Model, \emph{Procedia Computer Science}, \textbf{122}, 300-307.


\bibitem{Frey} Frey G., Manera M., Markandya A., \& Scarpa E. (2009), Econometric Models for Oil Price Forecasting: A Critical Survey, \emph{CESifo Forum, ifo Institute - Leibniz Institute for Economic Research at the University of Munich}, \textbf{10}(1), 29-44.

\bibitem{Grinblatt} Grinblatt M. \& Keloharju M. (2001), What makes investors trade?, \emph{J. Finance}, \textbf{56}, 589-616.


\bibitem{He} He X. J. (2018), Crude Oil Prices Forecasting: Time Series vs. SVR Models, \emph{Journal of International Technology and Information
Management}, \textbf{27}(2), 25-42.

\bibitem{Semere} Habtemicael S., Ghebremichael M., \& SenGupta I. (2019), Volatility and Variance Swap Using Superposition of the Barndorff-Nielsen and Shephard type L\'evy Processes, To appear in \emph{Sankhya B}, \url{https://doi.org/10.1007/s13571-017-0145-y}.


\bibitem{Issaka} Issaka, A. \& SenGupta, I. (2017), Analysis of variance based instruments for Ornstein-Uhlenbeck type models: swap and price index, \emph{Annals of Finance}, \textbf{13}(4), 401-434.

\bibitem{Issaka1}  Issaka, A. \& SenGupta, I. (2017), Feynman path integrals and asymptotic expansions for transition probability densities of some L\'evy driven financial markets, \emph{Journal of Applied Mathematics and Computing volume}, \textbf{54}, 159-182. 

\bibitem{Jiang} Jiang J. \& Tian W. (2018), Semi-nonparametric approximation and index options, \emph{Annals of Finance}, in press, \url{https://doi.org/10.1007/s10436-018-0341-4}. 


\bibitem{Kul} Kulkarni K.S. \& Sabarwal	 T. (2017), To what extent are investment bank-differentiating factors relevant for firms floating moderate-sized IPOs?, \emph{Annals of Finance}, \textbf{3} (3), 297–327. 

\bibitem{Li} Li X., Shang W., \& Wang S. (2019), Text-based crude oil price forecasting: A deep learning approach, \emph{International Journal of Forecasting}, \textbf{35} (4),  1548-1560.

\bibitem{NV}  Nicolato E. \& Venardos E. (2003), Option Pricing in Stochastic Volatility Models of the Ornstein-Uhlenbeck type, \emph{Math. Finance}, \textbf{13}, 445-466.

\bibitem{Pas} Pasiouras, F., Gaganis, C. \& Doumpos, M. (2007), A multicriteria discrimination approach for the credit rating of Asian banks, \emph{Annals of Finance}, \textbf{3}(3), 351-367. 

\bibitem{Michael} Roberts M. \& SenGupta I. (2019), Infinitesimal generators for two-dimensional L\'evy process-driven hypothesis testing, To appear in \emph{Annals of Finance }, \url{https://doi.org/10.1007/s10436-019-00355-y}.

\bibitem{ijtaf} SenGupta I. (2016), Generalized BN-S stochastic volatility model for option pricing, \emph{International Journal of Theoretical and Applied Finance}, \textbf{19}(02), 1650014 (23 pages).

\bibitem{SWW} SenGupta I., Wilson W., \&  Nganje W. (2019), Barndorff-Nielsen and Shephard model: oil hedging with variance swap and option, \emph{Mathematics and Financial Economics}, \textbf{13}(2), 209-226.

\bibitem{Sensoy} Sensoy A. \& Hacihasanoglu E. (2014), Time-varying long range dependence in energy futures markets, \emph{Energy Economics}, \textbf{46}(C), 318-327.


\bibitem{Tabak} Tabak B. M. \& Cajueiro D. O. (2007), Are the crude oil markets becoming weakly efficient over time? A test for time-varying long-range dependence in prices and volatility, \emph{Energy Economics}, \textbf{29}(1), 28-38.

\bibitem{SWW2} Wilson W., Nganje W., Gebresilasie S., \& SenGupta I. (2019), Barndorff-Nielsen and Shephard model for hedging energy with quantity risk, \emph{High Frequency}, \textbf{2} (3-4),  202-214.

\bibitem{Zhao} Zhao Y.,  Li J., \& Yu, L. (2017), A deep learning ensemble approach for crude oil price forecasting, \emph{Energy Economics}, \textbf{66}(C), 9-16.




\end{thebibliography}
\end{document}